%% file: main.tex
\documentclass[11pt,a4paper]{article}

\usepackage[T1]{fontenc}
\usepackage{newpxtext}
\usepackage{comment}
\usepackage{mathpazo}
\usepackage{amsmath, amsthm, amsfonts, amssymb, graphicx}
\usepackage{tikz}

\usepackage{dsfont}
\usepackage[margin=1truein]{geometry}
\usepackage[full]{complexity}
\usepackage{hyperref}
\usepackage{comment,bbm,enumitem}
\usepackage{xcolor}
\usepackage{algorithm, algpseudocode, algorithmicx}
\usetikzlibrary{matrix}

\usepackage[capitalize]{cleveref}
\crefname{observation}{observation}{observations}
\Crefname{observation}{Observation}{Observations}

\include{newcommand}

\usepackage{nicematrix}
\usepackage{mathtools}

\DeclarePairedDelimiterX{\Set}[2]{\{}{\}}{\,{#1}\,:\,{#2}\,}

\newcommand{\fsbr}[1]{\text{\textup{$\mathrm{<}\mkern-3.5mu\llap(\,#1\mathrm{>}\mkern-5.0mu\llap)$}}\,}

\DeclareMathOperator*{\nr}{nc-rank}
\DeclareMathOperator*{\qp}{quasi-poly}

\DeclareMathOperator*{\rank}{rank}

\let\span\relax
\DeclareMathOperator{\span}{span}
\DeclareMathOperator{\supp}{supp}

\newcommand{\ones}{\mathbf{1}}

\title{\textbf{Fractional Linear Matroid Matching is in quasi-NC}}
\author{Rohit Gurjar\thanks{ Indian Institute of Technology Bombay. E-mail: \href{mailto:rgurjar@cse.iitb.ac.in}{rgurjar@cse.iitb.ac.in}} \and Taihei Oki \thanks{The University of Tokyo, Tokyo, Japan. E-mail: \href{mailto:oki@mist.i.u-tokyo.ac.jp}{oki@mist.i.u-tokyo.ac.jp}} \and Roshan Raj \thanks{ Indian Institute of Technology Bombay. E-mail: \href{mailto:roshanraj@cse.iitb.ac.in}{roshanraj@cse.iitb.ac.in}}}
\date{}

\begin{document}
\maketitle  

\begin{abstract}
    The matching and linear matroid intersection problems are solvable in quasi-NC, meaning that there exist deterministic algorithms that run in polylogarithmic time and use quasi-polynomially many parallel processors.
    However, such a parallel algorithm is unknown for linear matroid matching, which generalizes both of these problems. In this work, we propose a quasi-NC algorithm for
    fractional linear matroid matching, which is a relaxation of linear matroid matching and commonly generalizes fractional matching and linear matroid intersection.
    Our algorithm builds upon the connection of fractional matroid matching to non-commutative Edmonds' problem recently revealed by Oki and Soma~(2023). As a corollary, we also solve black-box non-commutative Edmonds' problem with rank-two skew-symmetric coefficients.
\end{abstract}

\paragraph{Acknowledgements} We would like to thank Tasuku Soma for helpful comments and discussions.
Rohit Gurjar was funded by SERB MATRICS grant number MTR/2022/001009.
Taihei Oki was funded by JSPS KAKENHI Grant Number JP22K17853 and JST ERATO Grant Number JPMJER1903.
Roshan Raj was funded by PMRF grant. 


\input{introduction}\label{sec:introduction}
\input{Preliminaries}

\section{Isolating Weight Assignment for Fractional Matroid Parity}\label{sec:isolating}
In this section, we describe how we can construct an isolating weight assignment for fractional matroid parity with just the number of lines as input.
Our main theorem of this section is the following.

\begin{theorem}\label{thm:FLMMIsolatingWeight}
    There exists an algorithm that, given $m \in \mathbb{Z}_{+}$, outputs a set $\mathcal{W}\subseteq \mathbb{Z}_+^m$ of $m^{O(\log m)}$ weight assignments with weights bounded by $m^{O(\log m )}$ such that, for any fractional linear matroid matching polytope $P$ of $m$ lines, there exists at least one $w\in \mathcal{W}$ that is isolating for $P$, in time $\polylog (m)$ using $m^{O(\log m)}$ many parallel processors.
\end{theorem}

To prove \cref{thm:FLMMIsolatingWeight}, we use the following theorem given by Gurjar, Thierauf, and Vishnoi~\cite{DBLP:journals/siamcomp/GurjarTV21} that gives an isolating weight assignment for any polytope satisfying a certain property. For a face $F$ of a polytope, let $\mathcal{L}_F$ denote the lattice defined by
\[
    \mathcal{L}_F = \{v\in \mathbb{Z}^{m} \mid \text{$v= \alpha(x^1-x^2)$ for some $x^1,x^2\in F$ and $\alpha \in \F$}\}.
\]
 Let $\lambda(\mathcal{L})$ denote the length of the smallest non-zero vector of a lattice $\mathcal{L} \subseteq \mathbb{Z}^m$.  

\begin{theorem}[{\cite[Theorem 2.5]{DBLP:journals/siamcomp/GurjarTV21}}]\label{thm:IsolatingLatticeVertex}
    Let $k$ be a positive integer and $P \subseteq \mathbb{R}^m$ a polytope such that its extreme points are in $\{0,1/k,2/k,\dots,1\}^m$ and there exists a constant $c>1$ with
    \[|\{v\in \mathcal{L}_F \mid | v |  < c\lambda(L_F)\}| \leq m^{O(1)}\]
    for any face $F$ of $P$.
    Then, there exists an algorithm that, given $k$ and $m$, outputs a set $\mathcal{W}\subseteq \mathbb{Z}^m$ of $m^{O(\log km)}$ weight assignments with weights bounded by $m^{O(\log km )}$ such that there exists at least one $w\in \mathcal{W}$ that is isolating for $P$, in time $\polylog (km)$ using $m^{O(\log km)}$ many parallel processors.
\end{theorem}

It should be emphasized that the algorithm provided in \cref{thm:IsolatingLatticeVertex} requires only $k$ and $m$ and does not access the polytope $P$ itself.

\begin{remark}
    The result of \cite[Theorem 2.5]{DBLP:journals/siamcomp/GurjarTV21} gives an algorithm for a polytope with only integral extreme points. However, the proof can be easily modified to show the above theorem.
\end{remark}

Since a fractional matroid parity polytope $P$ is half-integral~\cite{GIJSWIJT2013509}, we just have to show that the number of near shortest vectors in $\mathcal{L}_F$ is polynomially bounded for any face $F$ of $P$ to apply \cref{thm:IsolatingLatticeVertex}.
Let $D_Fx=b_F$ be a system of equalities defining the affine space spanned by a face $F$.
Then, $\mathcal{L}_F$ is exactly the set of integral vectors   in the null space of $D_F$, i.e.,
\[\mathcal{L}_F=\{v\in \mathbb{Z}^m\mid D_Fv=0\}.\]
From \cref{thm:MaxFaceOfPolytope}, we can take $D_F$ such that
its entries are in $\{0,1,2\}$ and the sum of entries of every column is 2.
Now, we prove the following lemma, which directly implies \cref{thm:FLMMIsolatingWeight} using \cref{thm:IsolatingLatticeVertex}. 
\begin{lemma}\label{lem:nearShortestVectors}
     Let $D\in \{0,1,2\}^{p\times m}$ be a matrix such that the sum of entries of each column equals 2. Let $\mathcal{L}_D$ denote the lattice $\{v\in \mathbb{Z}^m \mid Dv=0\}$. Then, it holds that
     \[|\{v\in \mathcal{L}_D  : | v |< 2\lambda(\mathcal{L}_D)\}|\leq m^{O(1)}.\]
\end{lemma}

To prove \cref{lem:nearShortestVectors}, we introduce additional notions from \cite[Definition~3.1]{DBLP:journals/corr/SvenssonT17}.
Let $G$ be a multigraph with loops and $C=v_0\xrightarrow{e_0} v_1\xrightarrow{e_1}\dots \xrightarrow{e_{k-2}}v_{k-1}\xrightarrow{e_{k-1}}v_0$ a closed walk of even length in $G$ with edge repetition allowed. The \textit{size} of a closed walk $C$ denoted by $|C|$ is the number of edges in the walk. The \emph{alternating indicator vector}, denoted by $(\pm \mathbb{1})_C$, of $C$ is defined to be a vector $(\pm \mathbb{1})_C \coloneqq \sum_{i=0}^{k-1}(-1)^i\mathbb{1}_{e_i}$, where $\mathbb{1}_{e}\in \mathbb{R}^m$ is the elementary vector having 1 on position $e$ and 0 elsewhere. We say that $C$ is an \emph{alternating circuit} if its alternating indicator vector is non-zero. Note that for an alternating circuit $| (\pm \mathbb{1})_C | \leq |C|$.
For $x,y\in \mathbb{R}^m$, we say that $x$ is \emph{conformal} to $y$, denoted by $x \sqsubseteq y$, if $x_iy_i\geq 0$ and $|x_i|\leq |y_i|$ holds for all $i \in [m]$.

Now we are ready to prove \cref{lem:nearShortestVectors}.

\begin{proof}[{Proof of \cref{lem:nearShortestVectors}}]
  Let $G_D$ be a multigraph with vertex set $[p]$ and $m$ edges defined as follows.
  For every $e \in [m]$, the $e$th edge of $G_D$ is drawn between vertices $s$ and $t$ if $D[s,e]=D[t,e]=1$ for some $s, t \in [m]$ and it is a self-loop on vertex $s$ if $D[s,e]=2$ for some $s \in [m]$.
\begin{claim}\label{cl:latticeVectorDecomposition}
   For any $x\in \mathcal{L}_D$, there exists alternating circuits $C_1,C_2,\dots C_t$ in $G_D$ such that  $x=\Ac{C_1}+\Ac{C_2}+\dots +\Ac{C_t}$, $\Ac{C_i}\sqsubseteq x$, and $| \Ac{C_i} |=|C_i|$ for all $i \in [t]$.
\end{claim}
\begin{proof}
    First, we show that $\Ac{C}\in \mathcal{L}_D$ holds for any alternating circuit $C=v_0\xrightarrow{e_0} v_1\xrightarrow{e_1}\dots \xrightarrow{e_{k-2}}v_{k-1}\xrightarrow{e_{k-1}}v_0$.
    Let $D[i, *]$ denote the $i$th row of $A$.
    
    Then, by definition of $G_D$, we have
    \begin{equation} \label{eq:AClatticevector}
        D[i,*] \cdot \Ac{C}= \sum_{j=0}^{k-1}{(-1)}^j D[i,e_j]
    \end{equation}
    for all $i \in [p]$.
    Only those edges of $C$ can contribute to the above sum with one or both the endpoints as vertex $i$ in $G_D$.
    We can partition these edges such that each part contains edges that appear continuously in the walk. For example, $\xrightarrow{e_s}v_{s+1}\xrightarrow {e_{s+1}}\cdots \xrightarrow{e_{t-1}}v_{t}\xrightarrow{e_t}$ forms a part of the partition if $v_s\neq i, $ $i= v_{s+1}= v_{s+2} = \dots = v_{t}$ and $i\neq v_{t+1}$. This part contributes  ${(-1)}^s+2\sum_{j=s+1}^{t-1}{(-1)}^j + {(-1)}^t$ to the sum~\eqref{eq:AClatticevector}. If the number of self-loops is even, then $s$ and $t$ have different parties, and otherwise, they are the same.  Hence, each block contributes zero to the sum~\eqref{eq:AClatticevector}, implying $D\cdot \Ac{C} = \mathbf{0}$.
    
    We show \cref{cl:latticeVectorDecomposition} by decomposing given $x$ into alternating indicator vectors by an iterative algorithm described in \cref{algo:vectorDecomposition}. 

  
   \begin{algorithm}
		\caption{Decomposition of a lattice vector}\label{algo:vectorDecomposition}
		\begin{algorithmic}[1]
    \While {$| x | \neq 0$}
		\State $y\gets \mathbf{0}, j\gets 1$
        \State Let $e_0\in [m]$ such that $x_{e_0}>0$
        \State  $y_{e_0}\gets 1$  and let $e_0$th edge in $G_D$ be $\{v_0,v_1\}$
        \While{True}
        \If{$\exists e\in [m]$ such that the $e$th edge is $\{v_j,u\}$, $|x_e|>|y_e|$, and ${(-1)}^j x_e>0$}
            \State $y_e\gets y_e+{(-1)}^j$
            \State $e_j\gets e, \, v_{j+1}\gets u$
            \State $j\gets j +1$
        \Else
        \State Output $x\notin \mathcal{L}_D$\label{line:fail}
        \EndIf
        \If{$j \equiv 0 \pmod 2$ and $v_j=v_0$}\label{line:if}
            \State $x\gets x-y$
            \State output $y$ and exit inner while loop
        \EndIf
        \EndWhile
  \EndWhile
		\end{algorithmic}
  
	\end{algorithm}

Let $x$ denote the lattice vector of the current iteration. The inner while loop always terminates as $| y|$ increases each step and cannot exceed $| x|$. In each iteration, the algorithm follows a walk in $G_D$. If the current iteration ends at Line~\ref{line:if}, then we get a closed cyclic walk of even length, say $C$, such that $y$ is its alternating indicator vector. Since $| y|>0$, $C$ is indeed an alternating circuit. 
We further have $| y | = |C|$ and $y\sqsubseteq x$ by construction, implying $| x-y|=| x|-| y|<| x|$. 
Since $y$ is the alternating indicator vector of an alternating circuit, we have $y\in \mathcal{L}_D$, which in turn implies $x-y\in \mathcal{L}_D$.

Now, we show that \cref{algo:vectorDecomposition} never goes to Line~\ref{line:fail} if $x\in \mathcal{L}_D$, i.e., we always get an alternating circuit in each outer iteration.
Suppose to the contrary that at some iteration there is no edge $e \in \delta(v_j)$ such that $|x_e|>|y_e|$ and ${(-1)}^j x_e>0$, where $\delta(v_j)$ denote the set of edges adjacent to $v_j$ in $G_D$.
Let $P=v_0\xrightarrow{e_0} v_1\xrightarrow{e_1}\dots \xrightarrow{e_{j-1}}v_{j}$ be the walk obtained till now in the current iteration.
By the same argument as above, we can show that $D_iy=0$ holds for every $i\in [p]$ except $i\in \{v_0,v_j\}$. In fact, $D_{v_j}y >0 $ if $j$ is odd and otherwise $D_{v_j}y< 0$. Without loss of generality, let $j$ be odd. Let $\supp^+(x)$ and $\supp^-(x)$ denote the subsets $\{i\in [m]\mid x_i>0\}$ and $\{i\in [m]\mid x_i<0\}$, respectively. Then, we have \[D_{v_j}x= \sum\limits_{e\in \supp^+(x)\cap \delta(v_j)} D[v_j,e]|x_e|-\sum \limits_{e' \in \supp^-(x)\cap \delta(v_j)} D[v_j,e']|x_{e'}|=0\] and
  \begin{equation}\label{eq:Ay}
      D_{v_j}y= \sum\limits_{e\in \supp^+(y)\cap \delta(v_j)} D[v_j,e]|y_e|-\sum \limits_{e' \in \supp^-(y)\cap \delta(v_j)} D[v_j,e']|y_{e'}|> 0.
  \end{equation}
  By construction of $y$, we have $y\sqsubseteq x$. 
  Since there does not exist any $e \in \delta(v_j) $ such that $|x_e|>|y_e|$ and $-x_e>0$, $\supp^-(x)\cap \delta(v_j)$ is the same as $  \supp^-(y)\cap \delta(v_j) $ and $|x_{e'}|=|y_{e'}|$ holds for $e' \in \supp^-(y)\cap \delta(v_j)$. This contradicts equation~\eqref{eq:Ay} as $|x_e|\geq |y_e|$ for $e\in \supp^+(y)\cap \delta(v_j)$.
  Hence, the sum amount to $y$ in the current iteration is exactly the alternating indicator vector of an alternating circuit. Inductively, the algorithm successfully finds a decomposition of $x$.
\end{proof}

Now, we show that all the near-shortest vectors are alternating indicator vectors of an alternating circuit.

\begin{claim}\label{claim:shortest-alternating}
   Any lattice vector $x\in \mathcal{L}_D$ with $| x|<2\lambda(\mathcal{L}_D)$ is an alternating indicator vector $\Ac{C}$ of some alternating circuit $C$ in $G_D$ such that $| x| = |C|$.
\end{claim}
\begin{proof}
    Suppose to the contrary that the claim is not true.
    From \cref{cl:latticeVectorDecomposition}, there exists alternating circuits $C_1,C_2,\dots C_t$ with $t\geq 2$ such that $x=\Ac{C_1}+\Ac{C_2}+\dots +\Ac{C_t}$ with $\Ac{C_i}\sqsubseteq x$ and $| \Ac{C_i} |= |C_i|$ for all $i \in [t]$.
    We then have $| x|= \sum_{i=1}^t|\Ac{C_i}|\geq t\lambda(\mathcal{L}_D) \ge 2\lambda(\mathcal{L}_D)$, a contradiction. Hence, $ x = \Ac{C}$ for some alternating circuit $C$ with $| x| =|C|$.
\end{proof}

\Cref{claim:shortest-alternating} implies that  $\lambda(\mathcal{L}_D)$ is equal to the size of the smallest alternating circuit of $G_D$. It also implies that we just need a bound on the number of alternating indicator vectors that correspond to alternating circuits of size at most $2\lambda(\mathcal{L}_D)$ to prove \cref{lem:nearShortestVectors}.
The required bound on the number of near-shortest alternating circuits is given by the following theorem of Svensson and Tarnawski~\cite{DBLP:journals/corr/SvenssonT17}. ( The node-weight of an alternating circuit defined in \cite{DBLP:journals/corr/SvenssonT17} is same as its size for our case. ) 
\begin{theorem}[{\cite[Lemma~5.4]{DBLP:journals/corr/SvenssonT17}}]
    Let $G$ be a graph on $n$ vertices such that the size of the smallest alternating circuit is $\lambda$. Then, the cardinality
of the set
\[\{\Ac{C} : C \text{ is an alternating circuit in $G$ of size at most } 2\lambda\}\]
is at most $n^{17}$.
    \end{theorem}
    This completes the proof of Lemma \ref{lem:nearShortestVectors}.
 \end{proof}

\section{Finding Fractional Linear Matroid Matching via Isolation}\label{sec:algorithm}
In this section, we present an algorithm to find a fractional linear matroid matching to show \cref{thm:main}.
Let $\mathcal{W}$ be a set of weight assignments provided by \cref{thm:FLMMIsolatingWeight}.
Let $L = \{l_1, \dotsc, l_m\}$ be a set of lines with $l_i=\langle a_i,b_i \rangle$ for $i \in [m]$ and $A$ the associated matrix~\eqref{eq:A}.
Recall that, for $y\in \{0,1/2,1\}^m,$ the matrix $A^{\{2\}}(y)$ is defined as $\sum_{i=1}^m Y_i\otimes A_i$, where $Y_i=U_iU_i^T$ and $U_i$ is $2\times 2y_i$ matrix with indeterminates in its entries for $i\in [m]$.
Before presenting our algorithm, we show the following lemma.

\begin{lemma}\label{lem:DegAndFLMM}
    Let $w\in \mathbb{Z}^m$ be an isolating weight assignment with distinct weights for a fractional matroid parity polytope $P$.
    Let $t_{1,1},t_{1,2},t_{2,1},t_{2,2}$ be indeterminates and $\tilde{A}_w$ be the $2n\times 2n$ matrix obtained by substituting $t^{w_i}_{p,q}$ for the $(p,q)$ entry of $U_i$ in $A^{\{2\}}(\ones)$ for $i \in [m]$ and $p,q \in [2]$.
    Then, $ \pf(\tilde{A}_w)\neq 0$ if and only if there is a perfect fractional matroid matching.
    Moreover, it holds that \[\deg(\pf(\tilde{A}_w))=4 \max_{y\in P, \, |y|=\frac{n}{2}}w\cdot y,\]
    where $\deg$ means the total degree as a polynomial in four indeterminates $t_{1,1},t_{1,2},t_{2,1},t_{2,2}$.
\end{lemma}
\begin{proof}
    Firstly, we show the backward direction: $\pf(\tilde{A}_w)\neq 0\implies \pf(A^{\{2\}}(\ones))\neq 0 .$ This implies $A^{\{2\}}(\ones)$ has full rank, i.e., $2n$. From \cref{thm:ncrank}, we have \[\max_{y\in P}|y| = \frac{1}{2}\nr(A)= \frac{1}{4}\rank(A^{\{2\}}(\ones)) = \frac{n}{2}.\]
    Now, we show the other direction. From theorem \ref{thm:PfaffianExpansion},
    \[\pf(A^{\{2\}}(\ones))=\sum\limits_{\substack{z\in \{0,\frac{1}{2},1\}^m,\\|z|=n/2 }}\sum_{(J_1,\dots, J_m)\in \mathcal{J}(z)} \det([(U_1\otimes B_1)[J_1] \cdots (U_m\otimes B_m)[J_m]])\] 
    where $\mathcal{J}(z)$ is the family of $m$ tuples $(J_1,J_2,\dots, J_m)$ such that\[J_i=\begin{cases}
        \{1,2,3,4\} & \text{ if $z_i=1$, }\\
        \{1,2\} or \{3,4\} & \text{ if $z_i=1/2$,}\\
        \emptyset & \text{ if $z_i=0$.}
    \end{cases}\]
    Let $Q_z= \sum_{(J_1,\dots, J_m)\in \mathcal{J}(z)} \det([(U_1\otimes B_1)[J_1] \cdots (U_m\otimes B_m)[J_m]])$ and $\tilde{Q}_z$ denote $Q_z$ after the substitution. Now, we show the following claim.
    \begin{claim}
    \label{cl:degreeBound}
        For $z\in \{0,1/2,1\}^m$, if $\tilde{Q}_z\neq 0$ then, degree of any monomial in $\tilde{Q}_z$ is $4w\cdot z$.
    \end{claim}
    \begin{proof}

    Let $(J_1,J_2,\dots,J_m)\in \mathcal{J}(z)$ and $B$ be the matrix $[(U_1\otimes B_1)[J_1] \dots (U_m\otimes B_m)[J_m]]$ after the substitution. For $i\in [m]$ with $J_i\neq \emptyset$, each entry of $(U_i\otimes B_i)[J_i]$ has degree $w_i$ after the substitution. This also implies that for a fixed column of $B$, all the non-zero entries have the same degree. Hence, if $det(B)\neq 0$, the degree of any of its monomials is equal to the sum of the degrees of the non-zero entries of all the columns. This sum is equal to $ \sum_{j=1}^m |J_j|. w_j = 4w\cdot z$. Hence, degree of any monomial in $\tilde{Q}_z$ is $4w\cdot z$.
    \end{proof}
From \cite[Lemma 4.3]{DBLP:conf/soda/OkiS23}, $Q_z\neq 0$ implies that $z$ is a point in fractional matroid polytope. Let $z^*$ be the unique fractional matroid matching maximizing $w$. Hence, using claim \ref{cl:degreeBound} we can say that if $\tilde{Q}_{z^*}\neq 0,$ then $\pf(\tilde{A}_w)\neq 0$ with $\deg(\pf(\tilde{A}_w))=4w\cdot z^*$ as for any other $z$ in the polytope $\deg(\tilde{Q}_z)<\deg(\tilde{Q}_{z^*})$.

Now, we show that $\tilde{Q}_{z^*}\neq 0$. For $J=(J_1,\dots,J_m)\in \mathcal{J}(z)$, let $B_z(J)$ denote the matrix $[(U_1\otimes B_1)[J_1] \dots (U_m\otimes B_m)[J_m]]$, $Q_z(J)$ denote   $\det(B_z(J))$ and $\tilde{Q}_z(J)$ denote the four variate polynomial obtained by substitution. Let $\mathcal{L}_D(y)=\{i\in m\mid y_i=a\}$ where $y\in \{0,1/2,1\}^m$ and $a\in \{0,1/2,1\}$. Let $J=(J_1,\dots,J_m)\in \mathcal{J}(z^*)$ such that $J_i=\{1,2\}$ for $ i\in L_{1/2}(z^*)$. After substitution, for $i\in [m]$ with $J_i\neq \emptyset$, there are two columns in $B_z(J)$ such that each non zero entry in those columns has either $t_{1,1}^{w_i}$ or $t_{2,1}^{w_i}$. Hence, every monomial of $Q_{z^*}(J)$ is mapped to a monomial such that the sum of the powers of $t_{1,1}$ and $t_{2,1}$ is $\sum\limits_{i\in L_{1/2}(z^*)\cup L_1(z^*) } 2 w_i$. For any other tuple $J'\in \mathcal{J}(z^*)$,  the sum of the powers of $t_{1,1}$ and $t_{2,1}$ in the mapping is strictly less. Hence, $\tilde{Q}_{z^*}\neq 0 $ iff $\tilde{Q}_{z^*}(J)\neq 0$. 

Now, we show that $\tilde{Q}_{z^*}(J)\neq 0$. Without loss of generality, let $z^*_i=1$ for  $1\leq i\leq p$ and $z^*_i=1/2$ for $p< i\leq p+q$ and 0 otherwise such that $p=|L_1(z^*)|$ and $q=|L_{1/2}(z^*)|$. Then, 
\[
Q_{z^*}(J)=\det \begin{bmatrix}
            x^1_{1,1}B_1 & x^1_{1,2}B_1 & \dots  & x^p_{1,1}B_p & x^p_{1,2}B_p  & x^{p+1}_{1,1}B_{p+1}  & \dots &  x^{p+q}_{1,1}B_{p+q}   \\
            x^1_{2,1}B_1 & x^1_{2,2}B_1 & \dots    &        x^p_{2,1}B_p & x^p_{2,2}B_p  & x^{p+1}_{2,1}B_{p+1}  & \dots & x^{p+q}_{2,1}B_{p+q}          
        \end{bmatrix}.
\]
$z^*$ is an extreme point of the fractional matroid polytope. From  \cite[proof of Lemma 4.4]{DBLP:conf/soda/OkiS23}
 the above matrix can be transformed into the following matrix by changing the basis and then permuting the rows and columns:
  \begin{equation*}
     \begin{tikzpicture}[baseline=(current bounding box.center)]
\matrix (m) [matrix of math nodes,nodes in empty cells,right delimiter={]},left delimiter={[} ]{
 T & * & * & * & * & & * & *  \\
& U_{i_1} & U_{i'_1} &* &* & & * & * \\
& & & U_{i_2} & U_{i'_2} & & &    \\
   & & & & & & &  \\
  & & & & & & * & \\
  & & & & & & U_{i_q} & U_{i'_q}\\
} ;
\draw[thick][ dotted] (m-1-5)-- (m-1-7);
\draw[thick][ dotted] (m-2-5)-- (m-2-7);
\draw[thick][dotted] (m-2-8)-- (m-6-8);
\draw[thick][dotted] (m-2-5) -- (m-5-7);
\draw[thick][dotted] (m-2-7)--(m-5-7);
\draw[thick][dotted] (m-3-5)--(m-6-7);
\end{tikzpicture}.
 \end{equation*}
 These operations can be done by multiplying it with matrices in $GL(2n,\mathbb{R})$ from both sides. Note that this doesn't affect the fact that whether the determinant of the matrix is non-zero.
 Here, $U_j$ is $2\times 1$ matrix $[x^j_{1,1}, x^j_{2,1}]^T$ and $i_{k}\neq i'_{k}$ for $k \in [q]$ and belong to $\{p+1,\dots, p+q\}.$ The determinant is non-zero when $T$ and $[U_{i_k} U_{i'_k}] $ for $k \in [q]$ are non-singular.
 After substitution,
 \[[U_{i_k} U_{i'_k}] \mapsto \begin{bmatrix}
     t^{w_{i_k}}_{1,1} & t^{w_{i'_k}}_{1,1}\\
     t^{w_{i_k}}_{2,1} & t^{w_{i'_k}}_{2,1}
 \end{bmatrix}.\]
 Since weights of $w$ are distinct and $i_{k}\neq i'_{k}$, the matrix is non-singular. The only thing left is to show that $T$ is non-singular after the substitution. From \cite[proof of Lemma~4.4]{DBLP:conf/soda/OkiS23}, $T$ is a $4p\times 4p$ matrix that can be defined as follows:
 \[T=\begin{bNiceMatrix}
     x^1_{1,1}a'_1 & x^1_{1,1} b'_1 & x^1_{1,2}a'_1  & x^1_{1,2}b'_1 & \Cdots & x^p_{1,1}a'_p & x^p_{1,1} b'_p & x^p_{1,2}a'_p  & x^p_{1,2}b'_p\\
     x^1_{2,1}a'_1 & x^1_{2,1} b'_1 & x^1_{2,2}a'_1  & x^1_{2,2}b'_1 & \Cdots & x^p_{2,1}a'_p & x^p_{2,1} b'_p & x^p_{2,2}a'_p  & x^p_{2,2}b'_p
 \end{bNiceMatrix}.\]
 Here, $a'_i,b'_i\in \mathbb{F}^{2p}$ such that the matrix $D=[a'_1, b'_1, \dots , a'_p,b'_p]$ is non-singular. Let $T'$ denote the matrix $T$ after the substitution and $W=2\sum_{i=1}^pw_i$. By putting $t_{1,2}$ and $t_{2,1}$ to 0 in $T'$, the determinant of $T'$ is ${\det(D)}^2t_{1,1}^Wt_{2,2}^W$ which is non-zero. In conclusion, we showed
 \[\tilde{Q}_{z^*}(J)\neq 0 \implies \tilde{Q}_{z^*} \neq 0 \implies \pf(\tilde{A})\neq 0, \]
 as required.
\end{proof}
Now, we give \cref{algo:MainAlgo} to find a fractional matroid matching. Before that, here is an observation that we will use to get isolating weight assignments with distinct weights.
\begin{observation}\label{obs:DistinctWeights}
    Let $w\in \mathbb{Z}^m$ such that there is a unique $z^*$ in the fractional matroid matching polytope maximizing $w$ and $N$ be an integer $m^2$. Then, $w'$ with $w'_i=Nw_i+i$ is maximized uniquely by $z^*$ and $w'_i\neq w'_j $ for $i\neq j$.
\end{observation}
The above statement shows that we can construct an isolating weight assignment with distinct weights by a polynomial blow-up in the weights. From now on, we assume that the weight assignments have distinct weights.

For a weight assignment $w\in \mathbb{Z}^m$ and $e\in [m]$, let $w^e$ define a new weight assignment on $[m]$ with $w^e_i=4w_i+1$ if $i=e$ otherwise $4w_i$. For a given fractional linear matroid matching instance given as $m$ rank 2 skew-symmetric $n\times n$ matrices $A_i$ for $ i\in [m]$ and a  weight assignment $w$ on $[m]$ and a vector $v\in \{0,1/2,1\}^m$, let $$\tilde{A}_w(v)= \sum_{i=1}^m V_i\otimes A_i $$ where $V_i$ is the $2\times 2$ zero matrix if $v_i$ is $0$ otherwise $V_i=T_iT_i^T$. Here, $T_i$ is $\begin{bmatrix}
    t^{w_i}_{1,1} & t^{w_i}_{1,2}\\
    t^{w_i}_{2,1} & t^{w_i}_{2,2}
\end{bmatrix}$  if $v_i=1$ and $[t^{w_i}_{1,1}\:  t^{w_i}_{2,1}]^T$ otherwise. Let $\ones$ and $\mathbf{0}$ denote the $m$-dimensional vectors with all ones and all zeros.

 \begin{algorithm}
		\caption{Quasi-NC algorithm to find a perfect fractional linear matroid matching}
		\textbf{Input:} $A_i \in \mathbb{F}^{n\times n} \forall i\in [m]$ such that each $A_i$ is  rank 2 skew-symmetric.\\
		\textbf{Output:} A perfect fractional matroid matching.
    \label{algo:MainAlgo} 
		\begin{algorithmic}[1]
         \State $y\gets \mathbf{0}$
			\State Compute a family of weight assignments $\mathcal W$ as promised by Theorem~\ref{thm:FLMMIsolatingWeight}.
            \label{step:W}
            \PFor{$w\in\mathcal W$} 
            \State  $W\gets \deg(\det(\tilde{A}_w(\ones)))$.  \label{step:WDet}
            \If{$W>0$}
                \PFor{$e\in [m]$ }
                \State $W^e\gets \deg(\det(\tilde{A}_{w^e}(\ones)))$
                \State $ 
                  y_e \gets 
                  \begin{cases}
                    1 & \text{if } W^e=4W+8\\
                    1/2 & \text{if } W^e=4W+4\\
                    0 & \text{otherwise.}   
                  \end{cases}  
                  $
                \EndPFor
            \If{$|y|=n/2$ and $\det(\tilde{A}_{w}(y))\neq 0$}\label{line:is-FPM}
               \State Output $y$
               \EndIf
            \EndIf
            \EndPFor
       \State Output "No perfect fractional matroid matching exists."
		\end{algorithmic}
	\end{algorithm}

\begin{remark}
   The \cref{algo:MainAlgo} can be modified to solve the weighted fractional linear matroid matching in quasi-NC as well if the weights are $\qp(n)$. Since any subface of the maximum weight face of the fractional linear matroid matching polytope is again a face of the polytope, from \cref{thm:FLMMIsolatingWeight} we can construct a set of weight assignments $\mathcal{W}$ such that one of them isolates a maximum weight fractional matroid matching. For input weight assignment $v \in \mathbb{Z}_+^m,$ and $N>n\max_{w\in \mathcal{W},i\in [m]}w_i$, let $\mathcal{W}_v=\{Nv+w\mid w\in \mathcal{W}\}$. In the third step of \cref{algo:MainAlgo},  we choose weight assignments from $\mathcal{W}_v$ instead of $\mathcal{W}$  to find a maximum weight matching. 
\end{remark}
\paragraph{Proof of Correctness.} Let $P$ be the fractional linear matroid matching polytope for lines $\langle a_i,b_i\rangle$ for $i\in [m]$ such that $A_i=a_i\wedge b_i$. Suppose there doesn't exist a perfect fractional matroid matching, then from theorem \ref{thm:ncrank}\[
\rank(A^{\{2\}}(\mathbf{1}))= \rank(A^{\{2\}}) = 4\cdot \max_{z\in P}|z| < 2n.
\] Since $A^{\{2\}}(\mathbf{1})$ is not full rank, for any substitution of indeterminates $\det(A^{\{2\}}(\mathbf{1}))=0$. Hence, the algorithm doesn't return any perfect fractional matroid matching.

Now, suppose there exists a perfect fractional matroid matching. First, we show that if the algorithm returns $y\in \{0,1/2,1\}^m$ with $|y|=n/2$, then $y$ is indeed a perfect fractional matroid matching. Since $\det(\tilde{A}_{w}(y))\neq 0$, $\rank(A^{\{2\}}(y))=2n$ that is equal to the rank of $A^{\{2\}}$. Hence, from \Cref{thm:rankA2Y}, there exists $z\in \{0,1/2,1\}^m$ such that $z$ is a perfect fractional matroid  matching with $z\leq y$. Since $|y|=n/2$ and $y\in \{0,1/2,1\}^m$, we have $y=z$. Hence, $y$ is a perfect fractional matroid matching. 

Now, we show that we get a perfect fractional matroid matching for at least one $w\in\mathcal{W}$. From \cref{thm:FLMMIsolatingWeight}, there exists an isolating weight assignment for $P$ in $\mathcal{W}$.  From \Cref{obs:DistinctWeights}, we can assume that the weights of weight assignment $w$ are distinct for all $w\in \mathcal{W}$. Let $w^*\in \mathcal{W}$ be a weight assignment such that $z^*\in P$  uniquely maximizes $w^*$. From \cref{lem:DegAndFLMM} and the fact that the determinant of a skew-symmetric matrix is the square of its Pfaffian, for $w^*$,   $W$ in \cref{algo:MainAlgo} satisfies \[W=8 \max_{x\in P,|z|=\frac{n}{2}}w^*\cdot z= 8 w^*\cdot z^*.\]  Also, we have
\begin{equation}\label{eq:We} w^e\cdot z^*=\begin{cases}
    4w^*\cdot z^* + 1 & \text{ if } z^*_e=1,\\
    4w^*\cdot z^* + \frac{1}{2} & \text{ if } z^*_e=\frac{1}{2},\\
    4w^*\cdot z^*  & \text{ if } z^*_e=0.
\end{cases}\end{equation}
Let $z\in \{0,1/2,1\}^m$ be any other fractional matroid matching. Since weights of $w^*$ are positive integers and $w^*$  is an isolating weight assignment for half-integral polytope $P$, $w^*\cdot z^*-w^*\cdot z\geq 1/2$. Hence, $w^e\cdot z^*-w^e \cdot z\geq 1$ for $e\in [m]$. This implies $z^*$ uniquely maximizes $w^e$ for $e\in [m]$. Hence, using equation~\eqref{eq:We} and \cref{lem:DegAndFLMM}, $y=z^*$ for weight assignment $w^*$. Since $|z^*|=n/2$, from \cref{thm:PfaffianExpansion}, we get the following expansion of $\pf(\tilde{A}_{w^*}(z^*))$ after the substitution in $A^{\{2\}}(z^*)$,
\begin{equation*}
        \pf(\tilde{A}_{w^*}(z^*))=\sum_{(J_1,\dots, J_m)\in \mathcal{J}^{z^*}(z^*)} \det([(V_1\otimes B_1)[J_1] \cdots (V_m\otimes B_m)[J_m]]),
    \end{equation*}
    where $\mathcal{J}^{z^*}(z^*)$ is the family of $m$ tuples $(J_1,J_2,\dots, J_m)$ such that\[J_i=\begin{cases}
        \{1,2,3,4\} & \text{ if $z^*_i=1$, }\\
        \{1,2\} & \text{ if $z^*_i=1/2$,}\\
        \emptyset & \text{ if $z^*_i=0$.}
    \end{cases}\]
In the proof of \cref{lem:DegAndFLMM}, we have already shown that $\pf(\tilde{A}_{w^*}(z^*))$ ( $\tilde{Q}_{z^*}(J)$ ) is not equal to zero  if $z^*$ is an extreme point of $P$ and $z^*$ uniquely maximizes $w^*$. Hence, the algorithm outputs $z^*$ for weight assignment $w^*$.

\paragraph{Time Complexity.}
From \cref{thm:FLMMIsolatingWeight}, in $\polylog(m)$ time using $m^{O(\log m)}$ parallel processors, we can construct a set of $m^{O(\log m)}$ weight assignments $\mathcal{W}$ such that for $w\in \mathcal{W}$ and $i\in [m]$, $w_i$ is bounded by $m^{O(\log m)}$  and there exists an isolating weight assignment $w\in \mathcal{W}$ for the fractional linear matroid matching polytope. For $w\in \mathcal{W} $, the weights of $w^e$ for $e\in [m]$  are $m^{O(\log m)}$. Hence, the entries of $\tilde{A}_w(\ones)$ and $\tilde{A}_{w^e}(\ones)$ are polynomial in four variables with the individual degrees at most $m^{O(\log m)}$. Hence, from \cref{thm:DetQuasiNC}, the determinant of these matrices can be computed in $\polylog(m)$ time using $m^{O(\log m)}$ parallel processors. Hence, \cref{algo:MainAlgo} is a quasi-NC algorithm, and \cref{thm:main} has been proved. 

\section{Black-box Algorithm for Non-commutative Edmond's Problem with Rank-two Skew-symmetric Coefficients}\label{sec:black-box}

In this section, we explain that our results can be used to solve black-box non-commutative Edmonds' problem with rank-two skew-symmetric coefficients.

We first review the black-box and white-box settings for (commutative) Edmonds' problem and polynomial identity testing.
Recall that Edmonds' problem is to test the non-singularity of a given linear symbolic matrix $A = \sum_{i=1}^m A_i x_i$.
Since the non-singularity of $A$ is equivalent to $\det(A) = 0$, it is a special (and actually equivalent) class of \emph{polynomial identity testing} (PIT), which is to test if a given polynomial is zero.
In the white-box setting, a polynomial is given as an explicit formula such as an algebraic branching program or the coefficient matrices $A_1, \dotsc, A_m$ in Edmonds' problem, and in the black-box setting, we can only access to an oracle to evaluate a polynomial.
While a random substitution yields a randomized polynomial-time algorithm under the black-box setting, it is a long-standing open problem whether it can be derandomized, even for the white-box setting.

In non-commutative Edmonds' problem, a white-box algorithm requires explicit matrices $A_1, \dotsc, A_m$ as input same as the commutative problem.
In the black-box setting, recall from \cref{sec:introduction} that an algorithm is required to construct a non-commutative hitting set $\mathcal{H}$, that is, a set of $m$-tuples of square matrices over $\mathbb{F}$ such that for all $A=\sum_{i=1}^m x_iA_i$ with $A_i\in \mathbb{F}^{n \times n}$ for $i\in [m]$ and $\nr(A)=n$, there exists $(T_1,\dots, T_m)\in \mathcal{H}$ such that
$\det\left(\sum_{i=1}^m T_i\otimes A_i\right)\neq 0$.
Similar to the commutative setting, a random set $\mathcal{H}$ of polynomial size consisting of tuples of matrices of size at least $n-1$ works \cite{Derksen2017-jn}. Deterministically finding a hitting set, even of subexponential size, is still open \cite{DBLP:journals/focm/GargGOW20}. Nevertheless, some non-trivial results are known when we put some constraints on symbolic matrices. Gurjar and Thierauf \cite{GT17} construct a quasi-polynomial size hitting set when each $A_i$ has rank one. In this case, $\rank(A)=\nr(A)$. For a non-commutative algebraic formula with addition, multiplication, and inversion gates, there exists a symbolic matrix that has full non-commutative rank if and only if the non-commutative rational function computed by the formula is ``defined'' \cite{DBLP:journals/toc/HrubesW15}.  For such symbolic matrices, Arvind, Chatterjee and Mukhopadhyay \cite{arvind2023blackbox} construct a quasi-polynomial size hitting set. 

Here, we show the following corollary of \cref{thm:FLMMIsolatingWeight} and \cref{lem:DegAndFLMM}, which deterministically constructs a hitting set of quasi-polynomial size for the case where each $A_i$ is restricted to a rank-two skew-symmetric matrix, solving black-box non-commutative Edmonds' problem under this constraint in deterministic quasi-polynomial time.

\begin{corollary}
    For $m\in \mathbb{Z}_+$ and a field $\mathbb{F}$ of sufficient size, we can construct in time $\qp(m)$ a set $\mathcal{H}$ of size $\qp(m)$ that consists of $m$-tuples of $2 \times 2$ matrices over $\mathbb{F}$ such that, for all $n \times n$ symbolic matrix $A=\sum_{i=1}^m x_iA_i$ with rank-two skew-symmetric matrix $A_i$ and $\nr(A)=n$, there exists a tuple $(T_1,T_2,\dots T_m)\in \mathcal{H}$ that satisfies
    $\det \left(\sum_{i=1}^m T_i\otimes A_i \right) \ne 0$.
\end{corollary}
\begin{proof}
    Let $\mathcal{W}$ be the set of weight assignments $w:[m]\xrightarrow{}\mathbb{Z}_+$ from \cref{thm:FLMMIsolatingWeight} that can be constructed in $\qp(m)$ time. Let $D=\max \limits_{w\in \mathcal{W}, i\in [m]} w_i$ and $S\subseteq \mathbb{F}$ of size $2nD+1$. Then, for each $w\in {\mathcal{W}}$ and  $(a,b,c,d)\in S^4$, we put the tuple $(T_1,T_2,\dots ,T_m)$ in $\mathcal{H}$ such that $T_i=V_iV_i^T$ where,
    \[V_i=
    \begin{bmatrix}
     a^{w_i} & b^{w_i}\\
     c^{w_i} & d^{w_i}
 \end{bmatrix}\
    \]
    Let $A=\sum_{i=1}^m x_i A_i$ with $\nr(A)=n$ and $A_i$ as rank two skew-symmetric matrix for $i\in [m]$. Note that $n\leq 2m$. From \cref{thm:FLMMIsolatingWeight}, there exists an isolating weight assignment $w$ for fractional linear matroid matching polytope for $A$. From \cref{lem:DegAndFLMM}, $\tilde{A}_w$ is a $2n\times 2n$ matrix with entries as four variate polynomials with individual degree at most $D$ and $\det(\tilde{A}_w)\neq 0.$ It is a four variate polynomial with individual degree at most $2nD$. Hence, from the Schwartz--Zippel lemma, there exists a tuple $(T_1,T_2,\dots,T_m)\in \mathcal{H}$ corresponding to weight assignment $w$ such that $\det(\sum_{i=1}^m T_i\otimes A_i)$ is non-zero. The size of $\mathcal{H}$ is $|\mathcal{W}|\cdot |S|^4$, which is quasi-polynomially bounded in $m$ and hence can be constructed in time $\qp(m)$.
\end{proof}
\section{Conclusion} We showed that the problem of fractional linear matroid matching is in quasi-NC. The parallel complexity of linear matroid matching is still open. We also gave a black-box algorithm for non-commutative Edmond's problem for a symbolic matrix $A=\sum_{i=1}^mx_iA_i$ where each $A_i$ is a rank two skew-symmetric matrix. A natural question is whether we can extend our techniques to design a black-box algorithm when each $A_i$ has rank at most two or some larger positive integer constant. 
\bibliographystyle{alpha}
\bibliography{FractionalMatching}


\end{document}

%% file: newcommand.tex
\newcommand{\F}{\mathbb{R}}
\newtheorem{theorem}{Theorem}[section]
\newtheorem{lemma}[theorem]{Lemma}
\newtheorem{claim}[theorem]{Claim}
\newcommand{\Ac}[1]{(\pm \mathbb{1})_{#1}}
\DeclareMathOperator{\pf}{pf} 
\DeclareMathOperator{\sg}{sgn}

\newtheorem{corollary}[theorem]{Corollary}

\newtheorem{observation}[theorem]{Observation}

\algdef{SE}{PFor}{EndPFor}[1]{\textbf{for all} \(\mbox{#1}\) \textbf{do in parallel}}{\textbf{end for}}
\theoremstyle{definition}

\newtheorem{remark}[theorem]{Remark}

%% file: introduction.tex
\section{Introduction}
Algebraic algorithms play an important role in designing parallel algorithms because of the existence of highly parallelizable algorithms to perform basic matrix operations such as matrix multiplication and determinant computation. One of the earliest results in this direction can be attributed to Lov\'{a}sz~\cite{Lov79}, where the problem of testing the existence of a perfect matching in a graph is reduced to non-singularity testing of the corresponding Tutte matrix~\cite{10.1112/jlms/s1-22.2.107}. A (randomized) algorithm is called an (R)NC algorithm if it takes poly-logarithmic time and requires polynomially many parallel processors in terms of the input size. By assigning small integral values chosen randomly to the variables, we can efficiently test the non-singularity of a linear symbolic matrix, which is a matrix with linear forms in its entries. Using this, Lov\'{a}sz~\cite{Lov79} gave an RNC algorithm to decide the existence of a perfect matching in a graph. 

In the same work, he also reduced the more general \emph{linear matroid matching} problem (also known as the \emph{linear matroid parity} problem) to non-singularity testing of an associated symbolic matrix. In this problem, a set $L$ of two-dimensional vector subspaces, called \emph{lines}, of a vector space $\mathbb{F}^n$ over a field $\mathbb{F}$ is given as an input.
A subset $M$ of $L$ is called a \emph{matroid matching} if the dimension of the sum space of the lines in $M$ is $2|M|$, i.e., $\dim(\sum_{l \in M} l) = 2|M|$.
A matroid matching of size $n/2$ is called \emph{perfect}.
Then, the (perfect) linear matroid matching problem asks for a perfect matroid matching.
This problem generalizes the problems of finding a perfect matching in graphs as well as finding a common base of two linear matroids. For a set of lines $L=\{l_1,l_2,\ldots, l_m \}$, where $l_i =  \span(a_i,b_i) $ with $a_i,b_i\in \mathbb{F}^n$ for $i \in [m]$,  Lov\'{a}sz~\cite{Lov79} defined the following symbolic matrix $A$ in variables $x_1,x_2, \ldots, x_m$:
\begin{equation}\label{eq:A}
    A=\sum_{i=1}^m (a_i b_i^\top-b_i a_i^\top) x_i.
\end{equation}
Lov\'{a}sz showed that the determinant of $A$ is non-zero if and only if a perfect matroid matching exists, which in turn implies an RNC algorithm for linear matroid matching.

Lov\'{a}sz's above algorithms only solve the decision problems, i.e., only check the existence of a solution.
Some years later,  Karp, Upfal, and Wigderson~\cite{KUW86} and Mulmuley, Vazirani, and Vazirani~\cite{MVV87} gave RNC algorithms to find a perfect matching in graphs.
The algorithm of Mulmuley, Vazirani, and Vazirani~\cite{MVV87} introduces weights on the edges of the graph and then finds a maximum-weight perfect matching with respect to the assigned weights. They observed that this can be efficiently performed in parallel if the weight assignment is \emph{isolating}, i.e., there exists a unique maximum-weight perfect matching. They showed that a weight assignment with weights chosen randomly from a small set of integers is isolating with high probability in their famous \emph{isolation lemma}. Interestingly, the isolation lemma works not just for perfect matchings but for arbitrary families of sets.
Later on, Narayanan, Saran, and Vazirani~\cite{NSV94} obtained an RNC algorithm to find a common base of two linear matroids and a perfect matroid matching for linear matroids using the isolation lemma, along the same lines as~\cite{MVV87}.


In the deterministic setting, however, obtaining an NC algorithm is an outstanding open question even for the simplest of the cases, that is, deciding whether a bipartite graph has a perfect matching.
 Fenner, Gurjar, and Thierauf~\cite{FGT16} made significant progress in the direction of derandomization.
They provided a quasi-NC algorithm (i.e., deterministic polylogarithmic time with quasi-polynomially many parallel processors) for the bipartite matching problem, derandomizing the work of~\cite{MVV87}. Later, Svensson and Tarnwaski~\cite{ST17} and Gurjar and Thierauf~\cite{GT17} gave quasi-NC algorithms for perfect matching and linear matroid intersection, respectively. All these results go via constructing an isolating weight assignment for their respective problems in quasi-NC time. The deterministic construction of isolating weight assignments in~\cite{FGT16,ST17,GT17} relies on the linear program (LP) description of the associated polytope. For perfect matching and linear matroid intersection, the \emph{matching polytope} and the \emph{matroid intersection polytope} are defined as the convex hulls of indicator vectors of all the perfect matchings of the graph and all the common bases of two matroids, respectively.
Unfortunately, for linear matroid matching, an LP description of the \emph{linear matroid matching polytope} (the convex hull of indicator vectors of perfect matroid matchings) is still unknown, hindering us from building a quasi-NC algorithm for linear matroid matching.

In this paper, we work with a relaxation of linear matroid matching polytopes, called \emph{fractional linear matroid matching polytopes}, introduced by Vande Vate~\cite{DBLP:journals/jct/Vate92} and Chang, Llewellyn, and Vande Vate~\cite{DBLP:journals/dm/ChangLV01a}.
The linear fractional matroid matching polytope is defined for a set of lines $L=\{l_1,l_2,\dots,l_m\}$ as follows.
Let $E$ be a (multi)set of $2m$ vectors $\{a_1,b_1,\dots,a_m,b_m\}$ such that $l_i=\span( a_i,b_i)$ for $i \in [m]$. A subset $S$ of $E$ is called a \textit{flat} if $e\notin \span(S)$ holds for every $e\in E\setminus S$.
Then, a fractional matroid matching polytope is a collection of non-negative vectors $y\in \mathbb{R}^m$, called \emph{fractional matroid matchings}, such that
\begin{equation}
\label{eq:FLMMP}
    \sum_{i=1}^{m}y_i\cdot \dim(\span(S)\cap l_i)\leq \dim(\span(S))
\end{equation}
holds for all flats $S$ of $E$.
This polytope is a subset of ${[0, 1]}^m$ and the vertices are half-integral~\cite{GIJSWIJT2013509}. The polytope is a relaxation of the matroid matching polytope in the sense that the set of its integer vertices for $L$ is exactly the integer vertices of the matroid matching polytope for $L$ \cite{DBLP:journals/jct/Vate92}. The fractional matching polytope for a loopless graph and the matroid intersection polytope for two linear matroids coincide with the fractional matroid matching polytope for an appropriately chosen set of lines \cite{DBLP:journals/jct/Vate92}. The \emph{size} of a fractional matroid matching $y\in \mathbb{R}^m$ is defined as $\sum_{i=1}^m y_i$.
The (perfect) \emph{fractional linear matroid matching} problem is to find a \emph{perfect} fractional matroid matching, which is a fractional matroid matching  of size $n/2$.
A polynomial-time algorithm for fractional linear matroid matching was given by Chang, Llewellyn, and Vande Vate~\cite{DBLP:journals/dm/ChangLV01,DBLP:journals/dm/ChangLV01a} and this result was extended to weighted fractional linear matroid matching by Gisjwijt and Pap~\cite{GIJSWIJT2013509}.




Recently, Oki and Soma~\cite{DBLP:conf/soda/OkiS23} gave a relationship between fractional linear matroid matching and \emph{non-commutative Edmonds' problem}.
(Commutative) \emph{Edmonds' problem} is to test the non-singularity of a given linear symbolic matrix $A = \sum_{i=1}^m A_i x_i$.
Here, $A_i$'s are given $n\times n$ matrices over a field $\mathbb{F}$, and $A$ is regarded as a matrix over the rational function field $\mathbb{F}(x_1, \dotsc, x_m)$.
While the random substitution yields an efficient randomized algorithm, it is a long-standing open problem whether it can be derandomized.
In non-commutative Edmonds' problem, the variables $x_1, \dotsc, x_m$ are regarded as pairwise non-commutative, that is, $x_ix_j\neq x_jx_i$ if $i\neq j$.
Then, non-commutative Edmonds' problem asks to decide if the \textit{non-commutative rank}, denoted as $\nr(A)$, of $A$ is $n$ or not. Here, the non-commutative rank is defined as the rank of $A$ 
as a matrix over the \emph{free skew field}  $\mathbb{F}\fsbr{x_1,x_2,\dots,x_m}$, which is the quotient of the non-commutative polynomial ring $\mathbb{F}\langle x_1,x_2,\dots,x_m \rangle$~\cite{AMITSUR1966304}.
The non-commutative rank can also be characterized by the (commutative) rank of the \emph{blow-up} of $A$ defined as follows.
For $d\geq 1$, the $d$th-order blow-up of $A$, denoted by $A^{\{d\}}$, is the $dn\times dn$ linear symbolic matrix in $md^2$ variables given by
\[A^{\{d\}} = \sum_{i=1}^m X_i\otimes A_i,\] where $X_i$ is a $d\times d$ matrix with a distinct indeterminate in each entry for $i\in [m]$ and $\otimes$ denotes the Kronecker product. Then, $\nr(A)$ is equal to $\max_{d} \frac{1}{d}\rank(A^{\{d\}})$~\cite{DBLP:journals/cc/IvanyosQS18} and the inequality is attained for $d \ge n-1$~\cite{Derksen2017-jn}.
Unlike the commutative problem, non-commutative Edmonds' problem is known to be solvable in deterministic polynomial time~\cite{DBLP:conf/focs/GargGOW16,DBLP:journals/cc/IvanyosQS18,doi:10.1137/20M138836X} in the \emph{white-box} setting, i.e., the coefficient matrices $A_1, \dotsc, A_m$ are given as input.
The white-box setting is weaker than the \emph{black-box} setting, in which we need to construct a set $\mathcal{H}$, called a \emph{non-commutative hitting set}, of tuples of $m$ matrices such that for all $A=\sum_{i=1}^m x_iA_i$ with $\nr(A) = n$, there exists a tuple $(T_1,T_2,\dots,T_m)\in \mathcal{H}$ such that   $\det(\sum_{i=1}^m T_i\otimes A_i)$ is non-zero. Similar to the commutative setting, a hitting set of exponential size can be constructed trivially using Shwartz-Zippel Lemma and the polynomial dimension bounds of \cite{Derksen2017-jn}. For non-commutative Edmonds' problem under the black-box setting, even a subexponential-time deterministic algorithm is unknown~\cite{DBLP:journals/focm/GargGOW20}.  

Oki and Soma~\cite{DBLP:conf/soda/OkiS23} proved that the non-commutative rank of a linear symbolic matrix~\eqref{eq:A} with rank-two skew-symmetric coefficient matrices $A_i = a_i b_i^\top - b_i a_i^\top$ is $n$ if and only if there is a perfect fractional matroid matching in the corresponding fractional linear matroid matching polytope, analogous to Lov\'{a}sz' correspondence~\cite{Lov79} of commutative rank and matroid matching. 
They further showed that $\nr(A) = \frac12 \rank A^{\{2\}}$ holds for this $A$.
That is, the second-order blow-up is sufficient to attain $\nr(A)$ if the coefficient matrices are rank-two skew-symmetric.
Based on their results, Oki and Soma developed a randomized sequential algorithm for fractional linear matroid matching. Their result can be seen as a reduction from the search version to the decision version for the problem of fractional linear matroid matching. Then, the randomized algorithm decides the existence of a perfect fractional matroid matching, which can also be made to run in RNC.   However, it is not clear how to parallelize their reduction. 

\paragraph{Contribution.}
In this paper, we present a quasi-NC algorithm for fractional linear matroid matching, showing the following main theorem.

\begin{theorem}\label{thm:main}
    Fractional linear matroid matching is in quasi-NC.
\end{theorem}

Our algorithm comprises two parts: (i) constructing an isolating weight assignment and (ii) finding the unique maximum-weight fractional matroid matching. This generalizes the result of \cite{GT17} as the linear matroid intersection polytope coincides with the fractional linear matroid matching polytope for an appropriately chosen set of lines.

In the former part of our algorithm, we develop a quasi-NC algorithm to output a quasi-polynomially large set $\mathcal{W}$ of weight assignments such that at least one weight assignment is isolating for a given fractional linear matroid matching polytope.
To this end, we employ a parallel algorithm by Gurjar, Thierauf, and Vishnoi~\cite{DBLP:journals/siamcomp/GurjarTV21}.
For a polytope, they associate a lattice to each of its faces. Then, they show that if the lattice associated with each face of a polytope has polynomially bounded near-shortest vectors, then we can construct a quasi-polynomially large set of weight assignments such that at least one of them is isolating for the polytope (for details see \cref{thm:IsolatingLatticeVertex}). In their work, they show that a polytope for which each face lies in an affine space defined by a totally
unimodular matrix satisfies this property.  However, the faces of fractional linear matroid matching polytope do not lie in affine spaces defined by totally unimodular matrices. Hence, their work does not directly imply an isolating weight assignment for a fractional linear matroid matching polytope. To show that the lattice associated with each face of the fractional linear matroid matching polytope satisfies the above-mentioned property, we use the characterization of the faces of the fractional linear matroid matching polytope given by Gisjwijt and Pap~\cite{GIJSWIJT2013509}.

The latter part of our algorithm uses the non-commutative matrix representation of fractional linear matrix matching by Oki and Soma~\cite{DBLP:conf/soda/OkiS23}. For matching and linear matroid intersection, the parallel algorithm almost immediately follows using an isolating weight assignment. There is a one-to-one correspondence between the monomials of the pfaffian or determinant of the symbolic matrix and the perfect matchings or common bases of the matroids. After substituting the indeterminates of the symbolic matrix with univariate monomials with degrees as weights from the weight assignment, a monomial corresponding to a matching or a common base is mapped to a univariate monomial with degree equal to the weight of that matching or common base.  So, isolating a matching or a common base is equivalent to isolating a monomial of the pfaffian or  the determinant of the symbolic matrix, respectively. However, for fractional linear matroid matching, it is not immediately clear how an isolating weight assignment implies isolating a monomial in the pfaffian of the second-order blow-up of the symbolic matrix after the substitution. To show this, we use the expansion formula of the pfaffian of blow-up of the symbolic matrix and properties of extreme points of the fractional matroid matching polytope of \cite{DBLP:conf/soda/OkiS23}. 

Using our algorithm, we further derive an algorithm to construct a non-commutative hitting set of quasi-polynomial size for non-commutative Edmonds' problem where each $A_i$ is restricted to a rank-two skew-symmetric matrix, solving black-box non-commutative Edmonds' problem under this constraint in deterministic quasi-polynomial time.
See \cref{sec:black-box} for details.



\paragraph{Organization.}
The rest of this paper is organized as follows.
\Cref{sec:preliminaries} provides preliminaries on fractional linear matroid matching and some results and definitions from linear algebra.
The proposed algorithm is described in \cref{sec:isolating,sec:algorithm}.
\Cref{sec:isolating} describes how we can construct an isolating weight assignment and \cref{sec:algorithm} gives an algorithm for finding a perfect fractional linear matroid matching.
Finally, we present a black-box algorithm for non-commutative Edmonds' problem with rank-two skew-symmetric coefficients in \cref{sec:black-box}.


 

%% file: Preliminaries.tex
\section{Preliminaries and Notations}\label{sec:preliminaries}
We give the notations and definitions that we are going to use. Let $\mathbb{R}$, $\mathbb{Z}$, and $\mathbb{Z}_+$ represent the set of real numbers, integers, and non-negative integers, respectively. For a positive integer $n$, we denote by $[n]$ the set of integers $\{1,2,\dots,n\}$.

Let $\mathbb{F}$ denote the ground field of sufficient size.
For a vector $v \in \mathbb{F}^n$ and $i \in [n]$, $v_i$ denotes the $i$th component of $v$.
For vectors $a, b \in \mathbb{F}^n$, let $a \wedge b \coloneqq ab^\top - ba^\top$.
Let $A$ be an $n \times m$ matrix over $\mathbb{F}$.
For $S\subseteq [n]$ and $T\subseteq [m]$, $A[S,T]$ denote the submatrix of $A$ obtained by taking rows and columns of $A$ indexed by $S$ and $T$, respectively.
If $S$ is all the rows, we write $A[S, T]$ as $A[T]$.
If $S$ and $T$ are singletons, say $S = \{i\}$ and $T = \{j\}$, we simply write $A[\{i\}, \{j\}]$ as $A[i, j]$.
For vector spaces $V,W \subseteq \mathbb{F}^n$, we mean by $V\leq W$ that $V$ is a subspace of $W$. For vectors $a_1,  \dots, a_k\in \mathbb{F}^n$, $\langle a_1, \dotsc, a_k \rangle $ denotes the vector space spanned by $a_1, a_2, \dots, a_k$.

For two real vectors $x,y\in \mathbb{R}$, we mean by $x\leq y$ that $x_i\leq y_i$ holds for all $i\in [m]$.
The \emph{cardinality} of a non-negative vector $v\in \mathbb{R}^m$, denoted by $|v|$, is the $L_1$ norm $\sum_{i=1}^m |v_i|$.
Let $\ones$ denote the all-one vector of appropriate dimension.

\subsection{Linear Algebra Toolbox}
Here, we present some definitions and results of linear algebra that we use later.
The following result relates to the parallel computation of determinants, assuming that field operations take unit time. 
\begin{theorem}[\cite{BORODIN1983113}] \label{thm:DetQuasiNC}
    For an $n\times n$ matrix $A$ with entries as polynomials in a constant number of variables with individual degree at most  $d$, $\det(A)$  can be computed in $\polylog(n,d)$ time using $\poly(n,d)$ parallel processors.
\end{theorem}

For an $n\times m$ matrix $A={(a_{i,j})}_{i,j}$  and a $p\times q$ matrix $B$, the \emph{Kronecker product}, denoted by $A\otimes B$, is the $np\times mq$ matrix defined as
\[ 
    A \otimes B = \begin{bNiceMatrix}
        a_{1,1}B & \Cdots & a_{1,m}B \\
        \Vdots & \Ddots & \Vdots \\
        a_{n,1}B & \Cdots & a_{n,m}B
    \end{bNiceMatrix}.
\]
For a $2n\times 2n$ skew-symmetric matrix $A={(a_{i,j})}_{i,j}$, its \emph{Pfaffian} is defined as
\[\pf(A)=\frac{1}{2^nn!}\sum_{\sigma\in S_{2n}} \sg(\sigma)\prod_{i=1}^n a_{\sigma(2i-1),\sigma(2i)}\]
where $S_{2n}$ is the set of permutations $\sigma:[2n]\xrightarrow{}[2n].$ The Pfaffian is defined to be zero for skew-symmetric matrices with odd sizes, and it satisfies $(\pf(A))^2= \det(A)$.



\subsection{Fractional linear matroid matching}

Oki and Soma~\cite{DBLP:conf/soda/OkiS23} showed that the maximum cardinality of a fractional linear matroid matching is equal to half the non-commutative rank of the linear symbolic matrix $A$ given in~\eqref{eq:A}, analogous to Lov\'{a}sz' result~\cite{Lov79}.
They further showed that only second-order blow-up is sufficient to attain the non-commutative rank for $A$, summarized as follows.

\begin{theorem}[{\cite[Theorem~3.1]{DBLP:conf/soda/OkiS23}}]\label{thm:ncrank}
  Let $P$ be a fractional linear matroid matching polytope and $A$ the associated linear symbolic matrix defined by~\eqref{eq:A}.
  Then, it holds that
  \[
      \max_{y \in P} |y|
      = \frac12 \nr(A)
      = \frac14 \rank(A^{\{2\}}).
  \]
\end{theorem}

We use the following refinement of \cref{thm:ncrank}.
For a half-integral vector  $y\in \{0,1/2,1\}^m,$ let $$A^{\{2\}}(y)= \sum_{i=1}^m Y_i\otimes A_i,$$ where $Y_i=U_iU_i^T$ and $U_i$ is a  $2\times 2y_i$ matrix with indeterminates in its entries for $i\in[m]$.  In other words, $A^{\{2\}}(y)$ is the matrix obtained by substituting the $2\times 2$ symmetric matrix $Y_i$ of rank $2y_i$ into $X_i$ in the second-order blow-up $A^{\{{2}\}}$ for $i\in [m]$. 
Note that $\rank A^{\{2\}}(y) \le \rank A^{\{2\}}(\ones) \leq \rank A^{\{2\}}$ always holds.

\begin{theorem}[{\cite[Lemma~4.5]{DBLP:conf/soda/OkiS23}}] \label{thm:rankA2Y}
    Let $P$ be a fractional linear matroid matching polytope and $A$ the associated linear symbolic matrix defined by~\eqref{eq:A}.
    For $y\in \{0,1/2,1\}^m$, if $\rank(A^{\{2\}}(y))=\rank(A^{\{2\}})$ holds, then there exists $z\in \{0,1/2,1\}^m$ such that $z\leq y,$ $z\in P$, and $|z|=\max_{x\in P}|x|$.
    Conversely, if there exists such a point $z$ which is an extreme point of $P$, then $\rank(A^{\{2\}}(y))=\rank(A^{\{2\}})$ holds.
\end{theorem}

The matrix $ A^{\{2\}}(y)$ is skew-symmetric because each $Y_i$ is symmetric.
Oki and Soma also gave the following expansion formula for the Pfaffian of $A^{\{2\}}(y)$  that we will use later. Let $B_i=[a_i \:\: b_i]$ for $i\in [m]$.
\begin{theorem}[{\cite[Lemma 4.1]{DBLP:conf/soda/OkiS23}}]\label{thm:PfaffianExpansion}
    For $y\in \{0,1/2,1\}^m$, it holds that
    \begin{equation}
        \pf(A^{\{2\}}(y))=\sum\limits_{\substack{z\in {\{0,\frac{1}{2},1\}}^m,\\|z|=\frac{n}{2},\, z\leq y }}\sum_{(J_1,\dots, J_m)\in \mathcal{J}^y(z)} \det([(U_1\otimes B_1)[J_1] \cdots (U_m\otimes B_m)[J_m]]),
    \end{equation}
    where $\mathcal{J}^y(z)$ is the family of $m$ tuples $(J_1,J_2,\dots, J_m)$ such that\[J_i=\begin{cases}
        \{1,2,3,4\} & \text{ if $z_i=1$, }\\
        \{1,2\} or \{3,4\} & \text{ if $y_i=1,z_i=1/2$,}\\
        \{1,2\} & \text{ if $y_i=z_i=1/2$,}\\
        \emptyset & \text{ if $z_i=0$.}
    \end{cases}\]
\end{theorem}

The \textit{weighted fractional linear matroid matching} problem is to find a fractional matroid matching $y$ that maximizes $w\cdot y$ for a non-negative weight assignment $w:L\xrightarrow{}\mathbb{Z}_+$. Gijswijt and Pap \cite{GIJSWIJT2013509} gave a polynomial time algorithm for weighted fractional linear matroid matching. They also gave the following characterization for maximizing face of the  polytope with respect to a weight function. It is derived by showing the existence of a solution to the dual of the linear program of maximizing $w\cdot y$ over the inequalities~\eqref{eq:FLMMP} such that the support of the solution forms a chain.
\begin{theorem}[{see~\cite[proof of Theorem~1]{GIJSWIJT2013509}}] \label{thm:MaxFaceOfPolytope}
    Let $L=\{l_1,l_2,\dots,l_m\}$ be a set of lines with $l_i\leq \mathbb{F}^n$ and $w:L\xrightarrow{}\mathbb{Z}$ be a weight assignment on $L$. Let $F$ denote the set of fractional linear matroid matchings maximizing $w$ and $S\subseteq [m]$ such that every $y\in F$ has $y_e=0$ for $e\in S$. Then for some $k\leq n$, there exists a $k\times m$ matrix  $D_F$  and $b_F\in \mathbb{Z}^k$ such that
    \begin{itemize}
        \item the entries of $D_F$ are from $\{0,1,2\}$,
        \item the sum of the entries in any column of $D_F$ is exactly two, and
        \item a fractional matroid matching $y$ is in $F$ if and only if $y_e=0 \text{ for } e\in S$ and $D_Fy=b_F$.
    \end{itemize}
\end{theorem}